\documentclass[journal,draftclsnofoot,onecolumn]{IEEEtran}
\usepackage{amsmath}
\usepackage{amssymb}
\usepackage{amsthm}
\usepackage{color}
\usepackage{url}
\usepackage{comment}
\usepackage{eucal}
\usepackage{xcolor}
\usepackage{mathrsfs}
\usepackage{cite}
\usepackage{placeins}

\usepackage{caption} %
\usepackage{algorithm}
\usepackage{algpseudocode}
\usepackage{graphicx}  
\usepackage{float}
\usepackage{subfigure}
\usepackage{mathtools}
\DeclarePairedDelimiter{\set}{\{}{\}}

\usepackage{amsfonts,bm}
\renewcommand{\mathbf}[1]{{\bm{#1}}}
\renewenvironment{proof}{\begin{IEEEproof}}{\end{IEEEproof}}


\DeclarePairedDelimiter\norm{\lVert}{\rVert}
\DeclarePairedDelimiter\ceil{\lceil}{\rceil}

\newcommand{\cC}{{\mathcal{C}}}
\newcommand{\cX}{{\mathcal{X}}}

\newcommand{\Sph}{\mathbb{S}}
\renewcommand{\le}{\leqslant}
\renewcommand{\ge}{\geqslant}

\theoremstyle{plain}
\newtheorem{theorem}{\indent Theorem}

\newtheorem{lemma}[theorem]{\indent Lemma}
\newtheorem{proposition}[theorem]{\indent Proposition}

\newtheorem{construction}{Construction}

\newtheorem{problem}{Problem}

\theoremstyle{definition}

\theoremstyle{remark}
\newtheorem{remark}{\indent Remark}

\newcommand{\R}{\mathbb{R}}
\newcommand{\Z}{\mathbb{Z}}
\newcommand{\vol}{\operatorname{vol}}

\newcommand{\bbS}{\mathbb{S}}
\newcommand{\ip}[2]{\langle #1,#2\rangle}

\newcommand{\eqdef}{\triangleq}

\newcommand{\bldm}{{\mathbf{m}}}

\newcommand{\bldh}{{\mathbf{h}}}

\newcommand{\blds}{{\mathbf{s}}}
\newcommand{\bldu}{{\mathbf{u}}}
\newcommand{\bldv}{{\mathbf{v}}}
\newcommand{\bldx}{{\mathbf{x}}}
\newcommand{\bldy}{{\mathbf{y}}}

\newcommand{\Int}[1]{{\left[{#1}\right]}}

\allowdisplaybreaks
\begin{document}
\date{}
\title{Sharp Bounds and New Constructions for Single-Error Detection and Correction in Analog Codes} 	

\author{Hengzhuo Li, Zhengjie Jian, Xin Wang, and Hengjia Wei %
\thanks{This work was supported in part by the National Natural Science Foundation of China under Grant 12371523.}
\thanks{H. Li (leeker0626@outlook.com) and H. Wei (hjwei05@gmail.com) are with the School of Mathematics and Statistics, Xi'an Jiaotong University, Xi'an 710049, China. Z. Jian (xc169130@gmail.com) and X. Wang (xinw@suda.edu.cn) are with the Department of Mathematics, Soochow University, Suzhou 215005, Jiangsu, China.}

}

\maketitle

\begin{abstract}
We study single-error detection and correction for analog codes over $\R$. The key performance measures are the parameters $\Gamma_1(\cC)$ and $\Gamma_2(\cC)$, which quantify, respectively, the minimum separation required between large outlying errors that must be detected or located and the magnitude of tolerable perturbations. First, we prove that every real linear $[n,k]$ code $\cC$ satisfies
\[
\Gamma_1(\cC)\ge 2\left\lceil\frac{n}{n-k}\right\rceil.
\]
Moreover,  when $k=n-2$, we prove that every real linear $[n,n-2]$ code $\cC$ satisfies
\[
\Gamma_2(\cC)\ge \frac{1}{\sin^2(\pi/2n)}.
\]
Together, these two lower bounds settle all four open problems of Roth concerning the optimality of single-error-detecting and single-error-correcting analog codes. 
The proof of the first bound is based on a double-induction argument, while the proof of the second combines a zonotope-based geometric characterization of $\Gamma_2(\cC)$ with a cyclic sine-product inequality. In addition, we construct analog codes with higher fixed redundancy and show that, for every fixed $r\ge 2$, there exists a class of linear $[n,\ge n-r]$ codes over $\R$ such that
\[
\Gamma_2(\cC)\le O\left(n^{1+\frac{1}{r-1}}\right).
\]
This gives a new upper bound in the fixed-redundancy regime, which was not covered by previously known constructions.
\end{abstract}


   \section{Introduction}
\label{sec:introduction}

Analog computation has become an increasingly important paradigm in modern information processing systems, 
particularly in applications such as machine learning, signal processing, and in-memory computing~\cite{Zhaohan2026,NegiShubham2024,Chung2019,Guo2016,Yao2020MemristorCNN,Roth2022ITW}. 
In these settings, vector--matrix multiplication over the real field is a fundamental operation~\cite{Bucolo2021,SunZhong2022,HongMan2023,Zuo2025}. 
Unlike digital computation, however, analog computation is inherently affected by numerical perturbations, 
device imperfections, and occasional large-magnitude outliers. 
This motivates the study of error-correcting mechanisms over the real field that can distinguish small tolerable perturbations from large computational errors.

Roth introduced the framework of analog error-correcting codes for approximate real vector--matrix multiplication~\cite{ref1,ref2,ref3}. 
In this model, a linear code $\cC\subseteq \R^n$ is used to protect analog computation against two types of errors: 
small-magnitude errors whose entries are bounded by a tolerance level $\delta$, and outlying errors whose magnitudes exceed a larger threshold $\Delta$. 
The decoding goal is to locate a prescribed number of outlying errors while remaining robust to bounded tolerable perturbations. 
Thus, the performance of such a code is governed not only by the number of correctable or detectable outliers, 
but also by the required separation between the two error scales, namely the ratio $\Delta/\delta$. 
We refer to this ratio as the error-separation ratio.

For a linear code $\cC\subseteq \R^n$, Roth characterized the smallest admissible value of this ratio through a parameter $\Gamma_m(\cC)$, 
where $m=2\tau+\sigma$ is determined by the number $\tau$ of correctable errors and the number $\sigma$ of additional detectable errors. 
In other words, $\Gamma_m(\cC)$ represents the minimum error-separation ratio required for the corresponding analog decoding task.
Given the block length $n$, redundancy $r$, and parameter $m$, it is natural to ask how small $\Gamma_m(\cC)$ can be among all real linear $[n,n-r]$ codes. 
We denote this minimum value by $\Gamma_m(n,n-r)$.

In~\cite{ref1,ref2}, Roth proposed several coding schemes for single-error detection and single-error correction. Among other results, he showed that \[\Gamma_1(n,k)\le 2 \left\lceil\frac{n}{n-k}\right\rceil \textrm{ for every } 0<k<n,\]
and 
\[\Gamma_2(n,n-2)\le\frac{1}{\sin^2(\pi/2n)} \textrm{ for every } n>2.\]
However, it was left open whether these two upper bounds are tight. 
This led Roth to formulate four related problems~\cite[Problems~1--4]{ref1}. 
The first three problems concern single-error detection, while the fourth concerns single-error correction. 
Very recently, Jiang et al.~\cite{jiang2026tightlowerbound} solved \cite[Problem~2]{ref1}, and partially solved \cite[Problem~1]{ref1} under the condition that $n-k$ divides $k$.

In this paper, we focus on single-error detection and
correction, and solve all four open problems of Roth. We first prove that every real linear $[n,k]$ code $\cC$ satisfies
\[
\Gamma_1(\cC)\ge 2\left\lceil\frac{n}{n-k}\right\rceil.
\]
This solves Roth's first three problems~\cite[Problems~1--3]{ref1}. Compared with the recent work~\cite{jiang2026tightlowerbound}, our result applies in full generality, without any additional assumptions, and is obtained through a substantially simpler proof.

Moreover, we prove that every linear $[n,n-2]$ code $\cC$ over $\R$ with $n\ge 3$ satisfies
\[
\Gamma_2(\cC)\ge \frac{1}{\sin^2\frac{\pi}{2n}}.
\]
This gives an affirmative answer to \cite[Problem~4]{ref1}, and it follows that
\[
\Gamma_2(n,n-2)=\frac{1}{\sin^2\frac{\pi}{2n}}.
\]

We also study upper bounds for higher redundancy. 
Recently, Song and Cai~\cite{Song2026analog} constructed a class of linear $[n,n-3]$ codes satisfying
\[
\Gamma_2(\cC)\le
\frac{2n}{\sin \frac{\pi}{2\lceil\sqrt{(n-1)/2}\rceil}}.
\]
In this paper, we extend this direction by showing that, for every fixed integer $r\ge 2$, there exists a class of linear $[n,\ge n-r]$ codes over $\R$ such that
\[
\Gamma_2(\cC) \le O\left(n^{1+\frac{1}{r-1}}\right).
\]
It is worth noting that this exponent is consistent with the known low-redundancy constructions. 
For $r=2$, Roth's construction has order
\[
\frac{1}{\sin^2(\pi/2n)}=\Theta(n^2),
\]
which corresponds to the exponent $1+1/(r-1)=2$. 
For $r=3$, our bound gives order $O(n^{3/2})$, matching the order of the construction in~\cite{Song2026analog}. 
Thus, the bound above can be viewed as a higher-redundancy extension of the known constructions for $r=2$ and $r=3$.

We also briefly discuss related work. Roth~\cite{ref1} constructed a class of real linear codes satisfying $ \Gamma_2(\cC)\le 2\left\lceil \frac{2n}{r}\right\rceil $ for every $r$ such that $r(r-1)\ge n$. In the same paper and in~\cite{Roth2022ITW}, he proposed a construction based on spherical codes, which yields an infinite family of analog codes with $\Gamma_2(\cC)=O(n/\sqrt r) $ when $r=\Theta(\log n)$. 
Later, it was shown in~\cite{MyAdv} that the codes arising from this construction can also handle multiple errors. The same work further presented a construction of analog multiple-error-correcting codes based on sparse disjunct matrices. In the case $m=2$, this construction gives a class of linear codes satisfying $ \Gamma_2(\cC)\le \frac{2(\ell+1)n}{r}$, when $r=\Theta\bigl(n^{1/(\ell+1)}\bigr). $ More recently, geometric codes capable of handling multiple outliers were proposed and studied in~\cite{Roth2026,Zhu2026ANC}. In addition, algorithms for efficiently computing the $m$-height $h_m(\cC)$, which is related to $\Gamma_m(\cC)$ by $\Gamma_m(\cC)=2h_m(\cC)+2, $ were studied in~\cite{JiangAnxiao}, and several characterizations of $m$-heights were given in~\cite{Roth2026}. The known upper bounds on $\Gamma_2(n,n-r)$ for different ranges of $r$ are summarized in Table~\ref{tab:summary}.

\begin{table*}[ht]
  \caption{Upper Bounds on $\Gamma_2(n,n-r)$ for Different Ranges of $r$}
  \label{tab:summary}
\centering
  {\renewcommand{\arraystretch}{2.5}
   \everymath={\displaystyle}
  \begin{tabular}{cccl}
    \hline\hline
    Redundancy $r$     & $\Gamma_2(\cC)$  &Comments       & Reference\\
    \hline
    2
    & $\frac{1}{\sin^2\frac{\pi}{2n}}=O(n^2)$  
    &
    &  Proposition~11 in~\cite{ref1}                    \\
    3
    &$\frac{2n}{\sin \frac{\pi}{2\lceil\sqrt{(n-1)/2}\rceil}}=O(n\sqrt{n})$
    &
    &Theorem~3 in \cite{Song2026analog}\\
    Fixed $r\ge2$
    &$\frac{2n}{\sin\left(c_{r-1}n^{-1/(r-1)}\right)}=O(n^{1+\frac{1}{r-1}})$
    &Existence, $c_{r-1}
=
\left(\frac{(r-1)|\bbS^{r-1}|}{|\bbS^{r-2}|}\right)^{\frac{1}{r-1}}$
    & Theorem~\ref{thm:existence-explicit}\\
    Fixed $r\ge2$
    &$4n\left\lceil \left(\frac{n}{\kappa_{r-1}}\right)^{1/(r-1)}+\frac{\sqrt{r-1}}{2}\right\rceil$
    &$\kappa_d=\frac{\pi^{d/2}}{\Gamma(d/2+1)}$
    &Construction~\ref{cons:ball-grid}\\
    $\Theta(\log n)$
    & $O\left(\frac{n}{\sqrt{r}}\right)$ 
    &
    &  Proposition~5 in~\cite{Roth2022ITW}       \\
    $(\ell+1)q=\Theta(n^\frac{1}{\ell+1})$
    &$\frac{2(\ell+1)n}{r}$
    &$2 \le \lceil{q/\ell}\rceil$,
    $\ell \in \Z^+$, $q$ is a prime power, $n = q^{\ell+1}$
    &Corollary~16 in~\cite{MyAdv}\\
    $r\le n\le r(r-1)$
    & $2\left\lceil\frac{2n}{r}\right\rceil$
    &
    & Proposition~6 in~\cite{ref1}                             \\
    $n-2$
    &$2\frac{\cos\frac{\pi}{2n}}{\cos\frac{3\pi}{2n}}+2=O(1)$
    &
    &Theorem~III.2 in~\cite{Zhu2026ANC}\\
    \hline\hline
  \end{tabular}
  }
\end{table*}

The remainder of this paper is organized as follows.
Section~\ref{sec:pre} introduces the necessary preliminaries and defines the error-separation ratio $\Delta/\delta$ for analog codes.
Section~\ref{sec:detection} proves a sharp lower bound on $\Gamma_1(\mathcal C)$ for single-error detection.
Section~\ref{sec:cal} derives an explicit formula for $\Gamma_2(\mathcal C)$ in the redundancy-two case and establishes the corresponding lower bound for all $n\ge 3$.
The proof of this lower bound relies on a cyclic sine-product inequality, which is proved in Section~\ref{sec:proof}.
Finally, Section~\ref{sec:upperbound} presents upper bounds for analog codes with redundancy $r\ge 2$.

\section{Preliminaries}\label{sec:pre}
	For integers $\ell \le n$, we denote by $[\ell:n]$ the integer subset
$\set*{z \in \Z \,:\, \ell \le z < n}$.
We will use the shorthand notation $\Int{n}$ for $\Int{0:n}$,
and we will typically use
$\Int{n}$ to index the entries of vectors in $\R^n$.
	
	Given \(\delta, \Delta \in \mathbb{R}^+\), let
	\begin{equation}
			Q(n, \delta) 
			\triangleq 
			\{\boldsymbol{\epsilon} = (\epsilon_j) \in \mathbb{R}^n : \|\boldsymbol{\epsilon}\|_{\infty} \le \delta\}\nonumber
	\end{equation}
	be the set of all tolerable error vectors with threshold \(\delta\), where \(\|\boldsymbol{\epsilon}\|_{\infty}\) stands for the infinity norm, namely, \(\|\boldsymbol{\epsilon}\|_{\infty}= \max_{j \in [n]} |\epsilon_j|\). 
    We write $\|\cdot\|$ for the Euclidean norm.
    For \(\boldsymbol{e} = (e_j)_{j} \in \mathbb{R}^n\), let
	\[
	\text{Supp}_{\Delta}(\boldsymbol{e}) \triangleq \{j \in [n] : |e_j| > \Delta\}.
	\]
	
	In particular, \(\text{Supp}_0(\boldsymbol{e})\) is the ordinary support of \(\boldsymbol{e}\). We use \(w(\boldsymbol{e})\) to denote the Hamming weight of \(\boldsymbol{e}\). The set of all vectors of Hamming weight at most \(w\) in \(\mathbb{R}^n\) is denoted by \(B(n, w)\).
	
	Let \(\cC\) be a linear \([n, k]\) code over \(\mathbb{R}\). A decoder for \(\cC\) is a function \(D : \mathbb{R}^n \to 2^{[n]} \cup \{\text{``e''}\}\) which returns a set of locations of outlying errors or an indication ``e" that errors have been detected. Given \(\delta, \Delta \in \mathbb{R}^+\) and prescribed nonnegative integers \(\tau\) and \(\sigma\), we say that the decoder \(D\) corrects \(\tau\) errors and detects \(\sigma\) additional errors with respect to the threshold pair \((\delta, \Delta)\) if the following conditions hold for every 
	\[
	\boldsymbol{y} = \boldsymbol{c} + \boldsymbol{\epsilon} + \boldsymbol{e} \in \mathbb{R}^n,
	\]
	where \(\boldsymbol{c} \in C\), \(\boldsymbol{\epsilon} \in Q(n, \delta)\), and \(\boldsymbol{e} \in B(n, \tau + \sigma)\).
	
	\begin{itemize}
		\item[(D1)] If \(\boldsymbol{e} \in B(n, \tau)\), then \(D(\boldsymbol{y}) \neq \text{``e''} \subseteq \text{Supp}_0(\boldsymbol{e})\).
		\item[(D2)] If \(D(\boldsymbol{y}) \neq \text{``e''}\), then \(\text{Supp}_\Delta(\boldsymbol{e}) \subseteq D(\boldsymbol{y})\).
	\end{itemize}
	
	Let \(\boldsymbol{x} = (x_j)_{j \in [n]}\) be a nonzero vector in \(\mathbb{R}^n\) and let \(\pi\) be a permutation on \([n]\) such that
	\[
	|x_{\pi(0)}| \ge |x_{\pi(1)}| \ge \cdots \ge |x_{\pi(n-1)}|.
	\]
    
	Given an integer \(m \in [n]\), the \(m\)-height of \(\boldsymbol{x}\), denoted by \(h_m(\boldsymbol{x})\), is defined as
	\[
	h_m(\boldsymbol{x}) \triangleq \left|\frac{x_{\pi(0)}}{x_{\pi(m)}}\right|,
	\]
	and we formally define \(h_n(\boldsymbol{x}) \triangleq \infty\). For a linear code \({\cC} \neq \{0\}\) over \(\mathbb{R}\), its \(m\)-height, denoted by \(h_m({\cC})\), is defined by
	\[
	h_m({\cC}) \triangleq \max_{\boldsymbol{c} \in {\cC} \setminus \{0\}} h_m(\boldsymbol{c}).
	\]
	The minimum Hamming distance of \({\cC}\), denoted by \(d({\cC})\), can be related to \(h_m({\cC})\) by
	\begin{equation}\label{eq:dC}
			d({\cC}) = \min \{ m \in [n+1] : h_m({\cC}) = \infty \}. 
	\end{equation}

	\begin{theorem}[\cite{MyAdv,ref1}]\label{refRoth}
		Let \({\cC}\) be a linear \([n, k]\) code over \(\mathbb{R}\). There is a \((\tau, \sigma)\)-decoder for \(({\cC}, \Delta : \delta)\), if and only if
		\[
		\frac{\Delta}{\delta} \ge 2h_{2\tau + \sigma}({\cC}) + 2.
		\]
	\end{theorem}
	
	Recalling our definition of a decoder, 
	the decoding capability 
    of analog codes is characterized 
    not only
    by the number of correctable or detectable outlying errors
	(determined by the parameters $\tau$ and $\sigma$), 
	but also by the ratio $\Delta/\delta$. 
	Theorem~\ref{refRoth} provides a necessary and sufficient condition under which
	a given triple $(\tau, \sigma, \Delta/\delta)$ is attainable by a linear code ${\cC}$,
	in terms of the $m$-heights of ${\cC}$. In particular, by~(\ref{eq:dC}), 
	the inequality $d({\cC}) > 2\tau + \sigma$ is a necessary and sufficient condition for the existence of
	a $(\tau, \sigma)$-decoder for $({\cC}, \Delta : \delta)$, for some finite (yet
	sufficiently large) ratio $\Delta/\delta$.

	Theorem~\ref{refRoth} motivated Roth in~\cite{ref1} to define for every \(m \in [n+1]\) the expression
	\[
	\Gamma_m({\cC}) \triangleq 2h_m({\cC}) + 2,
	\]
	so that \(\Gamma_{2\tau + \sigma}({\cC})\) is the smallest ratio \({\Delta}/{\delta}\) for which there is a \((\tau, \sigma)\)-decoder for \(({\cC}, \Delta : \delta)\). 
	Equivalently, 
	\(\Gamma_{2\tau + \sigma}(\cC)\) is the smallest \(\Delta\) such that there is a \((\tau, \sigma)\)-decoder for \(({\cC}, \Delta : 1)\). 
    A natural question is then to determine $$\Gamma_m(n,k) \eqdef \min \set{ \Gamma_m(\cC): \cC \textrm{ is a linear $[n,k]$ code over $\R$} },$$ for given values of $n$, $k$ and $m$. 

    For the cases $m=1$ and $m=2$, which correspond to single-error detection and single-error correction, respectively, Roth proposed several code constructions and derived upper bounds on $\Gamma_m(n,k)$. 
These results can be summarized as follows.

\begin{theorem}[\cite{ref1}]
Let $k<n$ be positive integers. Then the following statements hold:
\begin{enumerate} 
    \item $\Gamma_1(n,k)\le 2 \ceil{\frac{n}{n-k}}$.
    \item $\Gamma_2(n,k)\le 2 \ceil{\frac{2n}{n-k}}$ if $n\le (n-k)(n-k-1)$.
    \item $\Gamma_2(n,n-2)\le \frac{1}{\sin^2(\pi/2n)}$ for $n>2$.
\end{enumerate}
\end{theorem}

However, apart from the trivial lower bound $\Gamma_m(\cC)\ge 4$, no nontrivial lower bounds were established in~\cite{ref1}. 
This led Roth to pose four problems concerning the optimality of the above upper bounds. 
In particular, Roth's first question is the following.

\setcounter{problem}{0}
\renewcommand{\theproblem}{\Alph{problem}}
\begin{problem}[{\cite[Problem~1]{ref1}}]\label{Pro-A}
    Identify the values of $k$ and $n$ for which every linear $[n,k]$ code $\cC$ over $\R$ satisfies
    \begin{equation}\label{eq:h1c}
        h_1(\cC)\ge\left\lceil\frac{k}{n-k}\right\rceil,
    \end{equation}
    or, equivalently, 
    $$\Gamma_1(\cC)\ge 2\left \lceil\frac{n}{n-k}\right\rceil.$$
\end{problem}
Roth's second question is the special case of Problem~A where $k=n-2$ and $n$ is even. Roth's third question gives a geometric reformulation of the second question under the additional assumption that each column has unit norm.

Roth's fourth problem concerns single-error correction.
\begin{problem}[{\cite[Problem~4]{ref1}}] Let $\cC(n)$ be the linear $[n,n-2]$ code over $\R$ defined by
\begin{equation}\label{construction}
        \cC(n)\triangleq
        \left\{
        (c_0,c_1,\ldots,c_{n-1})\in\R^n:
        \sum_{j\in[n]}c_j\omega^j=0
        \right\},
\end{equation}
where $\alpha=\pi/n$ and $\omega=e^{i\alpha}$ is a primitive $2n$-th root of unity. 

Does the code $\cC(n)$ have the smallest possible value of $\Gamma_2(n,n-2)$ among all linear $[n,n-2]$ codes over $\R$?
\end{problem}

In this paper, we first prove that, for every $0<k<n$, every linear $[n,k]$ code $\cC$ over $\R$ satisfies
    \[h_1(\cC)\ge \left\lceil\frac{k}{n-k}\right\rceil.\]
This settles \cite[Problems~1--3]{ref1}.
Moreover, we give an affirmative answer to \cite[Problem~4]{ref1}, by proving the matching lower bound
\begin{equation}\label{eq:mainbnd}
    \Gamma_2(n,n-2) \ge \frac{1}{\sin^2\frac{\pi}{2n}}.
\end{equation}
Consequently, all four open problems posed in~\cite{ref1} are resolved.

In addition, for fixed redundancy $r$, we present a new construction that yields the upper bound
\[
\Gamma_2(n,n-r)\le O\left(n^{1+\frac{1}{r-1}}\right).
\]
This provides an upper bound in the fixed-redundancy regime, which was not covered by the previously known constructions. In contrast, for other regimes of $r$, such as \[ r=\Theta(\log n),\qquad r=\Theta\bigl(n^{1/(\ell+1)}\bigr),\qquad r(r-1)\ge n, \] upper bounds on $\Gamma_2(n,n-r)$ have already been established in the literature; see Table~\ref{tab:summary}.

\subsection{A geometric characterization of $\Gamma_2$}
  To derive the bound \eqref{eq:mainbnd}, we need an alternative characterization of $\Gamma_2(\cC)$. 
For a $r\times n$ matrix $H$, let
\begin{equation}
	S_{H} \triangleq \set{ H\boldsymbol{x} : \boldsymbol{x} \in Q(n,1)}.\nonumber
\end{equation}
The set $S_H$ is the zonotope generated by the columns of $H$. 
It is a compact, convex, centrally symmetric subset of $\mathbb R^r$. 
In the special case $r=2$, if the columns of $H$ are pairwise nonparallel, then $S_H$ is a centrally symmetric polygon whose edges are parallel to the columns of $H$.

	\begin{proposition}[\cite{ref1}]\label{meanofD}
			Given a linear $[n,k,d\ge 3]$ code $\cC$ over $\mathbb{R}$, let
			$H = (\bldh_j)_{j\in[n]}$ be any $(n-k)\times n$ parity-check matrix of $C$
			and write $S = S_{H}$. Then $\Gamma_2(\cC)$ equals the smallest
			$\Delta \in \mathbb{R}^+$ such that
			for every distinct $j,j' \in [n]$ and every pair $(e,e') \in \mathbb{R}^2$
			such that $|e| > \Delta$, the translations
			\begin{equation}
				e \cdot \bldh_j + S \quad \text{and} \quad e' \cdot \bldh_{j'} + S \label{form1} 
			\end{equation}
			are disjoint; equivalently,
			\begin{equation}
				e' \cdot \bldh_{j'} \notin e \cdot \bldh_j + 2S.\label{form2}
			\end{equation}
	\end{proposition}

By Proposition~\ref{meanofD}, if the parity-check matrix $H$ contains two parallel columns, then $\Gamma_2(\cC)=\infty$. 
Hence, in studying finite values of $\Gamma_2(\cC)$, we may assume that the columns of $H$ are pairwise nonparallel. 
Moreover, permuting the columns of $H$ or multiplying any column by $\pm1$ does not change $\Gamma_2(\cC)$. 
Therefore, in the redundancy-two case, we may assume without loss of generality that the columns
\[
\boldsymbol h_0,\boldsymbol h_1,\ldots,\boldsymbol h_{n-1}
\]
are ordered by increasing polar angle and satisfy
\[
0\le \arg(\boldsymbol h_0)
<
\arg(\boldsymbol h_1)
<
\cdots
<
\arg(\boldsymbol h_{n-1})
<
\pi.
\]

Our proof of the lower bound for $\Gamma_2(n,n-2)$ relies on Proposition~\ref{meanofD} and the theorem below.

\begin{theorem}[Separation of an affine set and a convex set,{\normalfont\cite[Page~49]{ConvexBoyd}}]\label{thm:separ}
  Let $C$ be a convex set and let $D$ be an affine set.
    If $C$ and $D$ are disjoint, then there exist $\bldu\neq \boldsymbol{0}$ and $b\in\mathbb{R}$ such that
    \[
        \langle \bldu,\bldx\rangle \le b
        \quad \text{for all } \bldx\in C,
    \]
    and
    \[
        \langle \bldu,\bldx\rangle \ge b
        \quad \text{for all } \bldx\in D.
    \]
\end{theorem}

\subsection{Upper bound based on coherence}

Let $H$ be a real matrix with normalized columns, i.e., \[ \norm{\bldh_i}=1, \qquad i=0,\ldots,n-1. \] The coherence of $H$ is defined as \[ \rho(H)\eqdef \max_{i\ne j} \left|\ip{\bldh_i}{\bldh_j}\right|. \] We shall use the following bound due to Song and Cai~\cite[Theorem~1]{Song2026analog}. \begin{proposition}\label{prop:coh} If $\rho(H)<1$, then the code $\cC=\ker H$ satisfies \begin{equation}\label{eq:coherence-bound} \Gamma_2(\cC) \le \frac{2n}{\sqrt{1-\rho(H)^2}}. \end{equation} \end{proposition}

\section{Sharp lower bound for single-error detection}
\label{sec:detection}
In this section,
we address Problem~A by proving~\eqref{eq:h1c}.
To this end, we first establish the following key lemma.

\begin{lemma}\label{lal}
 Let \( t,r,N \) be positive integers. For any vectors \( \boldsymbol{v}_0, \boldsymbol{v}_1, \dots, \boldsymbol{v}_{N} \in \mathbb{R}^r \) with $N\ge tr$, there exists  \( i \in [N+1] \)  such that
$$
t\boldsymbol{v}_i=\sum_{j\ne i}\lambda_j\boldsymbol{v}_j,
$$
where $|\lambda_j|\le 1$ for $j\in[N+1]\setminus\{i\}$.
\end{lemma}

\begin{proof}
    We proceed by double induction on \( (t,r) \). 

When \( r = 1 \), then  \( \boldsymbol{v}_0, \boldsymbol{v}_1, \dots, \boldsymbol{v}_{t} \) are \( t+1 \) real numbers. Take \( \boldsymbol{v}_i \) as the one with the smallest absolute value, that is \( |\boldsymbol{v}_i| = \min\{|\boldsymbol{v}_0|, |\boldsymbol{v}_1|, \dots, |\boldsymbol{v}_{t}|\} \). 
If $\min|\bldv_i|=0$, then the conclusion holds by taking all $\lambda_j=0$. 
Otherwise, let 
$$
\lambda_j=\left\{\begin{array}{ll}
    \frac{v_i}{v_j}, & \textup{if}~ 0\le j\le t,j\ne i,  \\
    0, & \textup{if}~ j\ge t+1.
\end{array}\right.
$$
Then the conclusion holds.   

When \( t = 1 \), the vectors \( \boldsymbol{v}_0, \boldsymbol{v}_1,\dots,\boldsymbol{v}_N \) are linearly dependent. Thus, there exist constants \(c_0, c_1, \dots, c_N\), not all zero, such that
\[
\sum\limits_{i=0}^Nc_i\boldsymbol{v}_i=0.
\]
Take \( |c_i| \) to be the maximum among \( |c_0|, |c_1|,\dots, |c_N| \). Then we can rearrange the equation above as
\[
\boldsymbol{v}_i =\sum\limits_{j\ne i}\left( -\frac{c_j}{c_i}\right)\boldsymbol{v}_j,
\]
where \( |\lambda_j| = \left| -\dfrac{c_j}{c_i} \right| \le 1 \) for $j\in[N+1]\setminus\{i\}$. 

This completes the base case.

Assume the conclusion holds for \( (\le t-1,\le r) \) and \((\le t,\le r-1)\). We now consider the case for \( (t,r) \). If $\textup{rank}(\boldsymbol{v}_0, \boldsymbol{v}_1, \dots, \boldsymbol{v}_N)=w\le r-1$, there exists non-singular linear transformation $\mathcal{A}$, such that $\mathcal A\boldsymbol{v}_i=\begin{pmatrix}
    \boldsymbol{u}_i\\
    \boldsymbol{0}
\end{pmatrix}$ for $i\in[N+1]$, where $\boldsymbol{u}_i\in\R^w$. By the inductive hypothesis for \((t,\le r-1)\), there exists $i\in[N+1]$ such that
$$
t\boldsymbol{u}_i=\sum\limits_{j\ne i}\lambda_j\boldsymbol{u}_j,
$$
where $|\lambda_j|\le 1$. Thus, we have that
$$
t\mathcal{A}\boldsymbol{v}_i=\sum\limits_{j\ne i}\lambda_j\mathcal{A}\boldsymbol{v}_j.
$$
Applying \( \mathcal{A}^{-1} \) to both sides, we obtain
$$
t\boldsymbol{v}_i=\sum\limits_{j\ne i}\lambda_j\boldsymbol{v}_j.
$$
The conclusion holds for this case.

Thus we can assume that $\textup{rank}(\boldsymbol{v}_0, \boldsymbol{v}_1, \dots, \boldsymbol{v}_{N})=r$. Since the order of the vectors does not affect the conclusion, without loss of generality, among \(|\det(\boldsymbol{v}_{i_0}, \boldsymbol{v}_{i_1}, \dots, \boldsymbol{v}_{i_{r-1}})|\), we assume that \(\boldsymbol{v}_0, \boldsymbol{v}_1, \dots, \boldsymbol{v}_{r-1}\) attain its maximum. Define a non-singular linear transformation \( \mathcal{A} \) such that
   \[
   \mathcal{A}\boldsymbol{v}_i = \boldsymbol{e}_i, \quad i\in[r],
   \]
   where $\{\boldsymbol{e}_0,\boldsymbol{e}_1,\dots,\boldsymbol{e}_{r-1}\}$ are the standard basis vectors of $\R^r$. 
   
   For $j\ge r$,  let $\boldsymbol{u}_j=\mathcal{A}\boldsymbol{v}_j=(u_{j0},u_{j1},\dots,u_{j(r-1)})^T\in\R^r$. Since
   \[
   |\det(\mathcal{A}\boldsymbol{v}_{i_0}, \mathcal{A}\boldsymbol{v}_{i_1},\dots,\mathcal{A}\boldsymbol{v}_{i_{r-1}})| = |\det(\mathcal{A})| \cdot |\det(\boldsymbol{v}_{i_0}, \boldsymbol{v}_{i_1}, \dots, \boldsymbol{v}_{i_{r-1}})|,
   \]
   by the maximality of \( |\det(\boldsymbol{v}_0, \boldsymbol{v}_1, \dots, \boldsymbol{v}_{r-1})| \), for all \( j\ge r\) and $0\le l\le r-1$, we have that 
   \begin{align*}
       |u_{jl}|&=|\det(\boldsymbol{e}_0,\dots,\boldsymbol{e}_{l-1},\boldsymbol{u}_{j},\boldsymbol{e}_{l+1},\dots,\boldsymbol{e}_{r-1})|\\
       &=|\det(\mathcal{A})|\cdot|\det(\boldsymbol{v}_0,\dots,\boldsymbol{v}_{l-1},\boldsymbol{v}_{j},\boldsymbol{v}_{l+1},\dots,\boldsymbol{v}_{r-1})|\\
       &\le|\det(\mathcal{A})|\cdot|\det(\boldsymbol{v}_0,\boldsymbol{v}_1,\dots,\boldsymbol{v}_{r-1})|\\
       &=|\det(\boldsymbol{e}_0,\boldsymbol{e}_1,\dots,\boldsymbol{e}_{r-1})|=1.
   \end{align*}

   Since \(N-r \ge (t-1)r\), the inductive hypothesis for \((t-1,r)\)
applies to the vectors
\(\boldsymbol u_r,\boldsymbol u_{r+1},\ldots,\boldsymbol u_N\).
Therefore, there exists an index \(i\ge r\) such that
\begin{align}\label{equ1}
(t-1)\boldsymbol u_i 
=
\sum_{\substack{j\ne i,~ j\ge r}}
\lambda_j \boldsymbol u_j,
\end{align}
where \(|\lambda_j|\le 1\) for all \(j\ne i\) with \(j\ge r\).
Additionally, since $\{\boldsymbol{e}_0,\boldsymbol{e}_1,\dots,\boldsymbol{e}_{r-1}\}$ are the standard basis vectors, then
   \begin{align}\label{equ2}
       \boldsymbol{u}_i =\sum\limits_{l=0}^{r-1}u_{il}\boldsymbol{e}_l.
   \end{align}
 Combining $(\ref{equ1})$ and $(\ref{equ2})$, we get
   \[
   t\boldsymbol{u}_i=\sum\limits_{l=0}^{r-1}u_{il}\boldsymbol{e}_l+\sum_{j \neq i,~j\ge r} \lambda_j \boldsymbol{u}_j.
   \]
Applying \( \mathcal{A}^{-1} \) to both sides, since \( \mathcal{A}^{-1}\boldsymbol{e}_l = \boldsymbol{v}_l \) and \( \mathcal{A}^{-1}\boldsymbol{u}_j = \boldsymbol{v}_j \), we obtain
   \[
   t \boldsymbol{v}_i = \sum\limits_{l=0}^{r-1}u_{il}\boldsymbol{v}_l + \sum_{j \neq i,~j\ge r} \lambda_j \boldsymbol{v}_j.
   \]
By the bounds \( |u_{il}| \le 1 \) and \( |\lambda_j| \le 1 \), the conclusion holds for $(t,r)$.
\end{proof}

\begin{theorem}
        Let \( k \) and \( r \) be positive integers, and \( k = qr + s \), where \( q \) and \( s \) are non-negative integers and \( 1\le s \le r \). Let \( \cC \) be a linear $[k+r,k]$ code over $\R$. Then 
        $$
        h_1(\cC)\ge q+1.
        $$
    \end{theorem}

\begin{proof}
Let $H_{r\times (k+r)}$ be the parity-check matrix of $\cC$, $\boldsymbol{h}_0,\boldsymbol{h}_1,\dots,\boldsymbol{h}_{k+r-1}$ be the column vectors of $H$. Since
$$
k+r=(q+1)r+s\ge (q+1)r+1,
$$
by Lemma $\ref{lal}$ there exists $i$ such that
$$
(q+1)\boldsymbol{h}_i=\sum_{j\ne i}\lambda_j\boldsymbol{h}_j,
$$
where $|\lambda_j|\le 1$ for $j\in[k+r]\setminus\{i\}$. Thus
$$
\boldsymbol{x}=(-\lambda_0,\dots,-\lambda_{i-1},q+1,-\lambda_{i+1},\dots,-\lambda_{k+r-1})\in\cC.
$$
By the definition of $h_1(\cC)$, we have
$$
h_1({\cC}) = \max_{\boldsymbol{c} \in {\cC} \setminus \{0\}} h_1(\boldsymbol{c})\ge h_1(\boldsymbol{x})=\frac{q+1}{\max_{j\ne i}\{|\lambda_j|\}}\ge q+1. 
$$
    
\end{proof}

\begin{remark}
    \begin{enumerate}
        \item When $r=1$, our proof is essentially identical to that of Proposition $5$ in \cite{ref1}.
        \item When $r\mid k$, we provide a new proof that is completely different from and simpler than that in \cite{jiang2026tightlowerbound}. Moreover, our method also applies to the general case when $r\nmid k$.
    \end{enumerate}
\end{remark}

\section{Lower bound on $\Gamma_2(n,n-2)$}
\label{sec:cal}
Let $\cC$ be a linear $[n,k]$ code with parity-check matrix $H$. 
Denote $S\eqdef \set*{H\boldsymbol{x}  :  \boldsymbol{x} \in Q(n,1) } $. For any two distinct indices $i,j \in \Int{n}$, let  $\Delta_{(i,j)}$ denote the smallest
			$\Delta \in \mathbb{R}^+$ such that
			for  every pair $(e,e') \in \mathbb{R}^2$
			with $|e| > \Delta$, the translated sets
			\begin{equation}
				e \cdot \bldh_i + S \quad \text{and} \quad e' \cdot \bldh_{j} + S
			\end{equation}
			are disjoint; see, for example, Figure~\ref{fig:placeholder}.
            \begin{figure}
                \centering
                \includegraphics[width=0.4\linewidth]{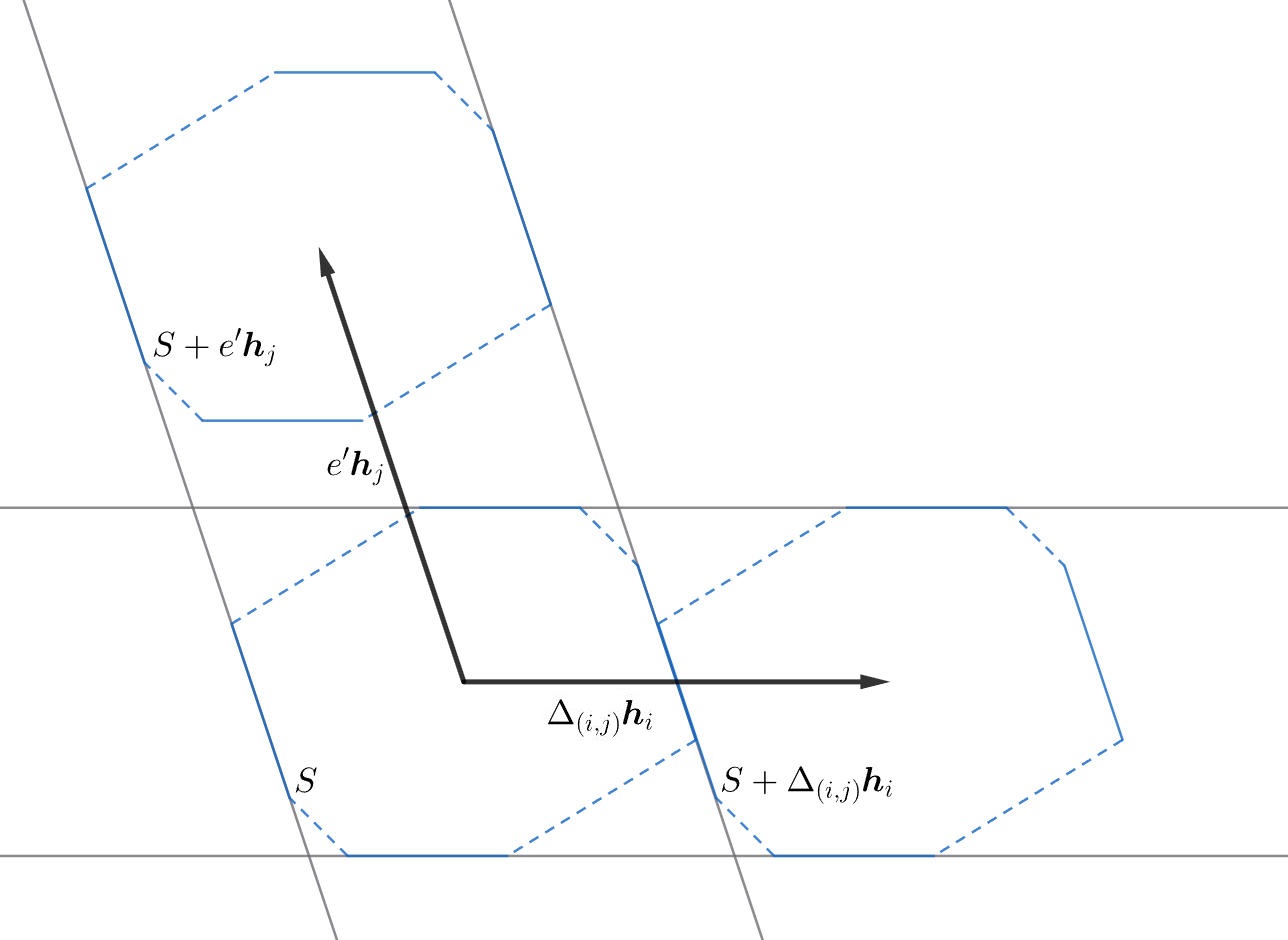}
                \caption{Illustration of $\Delta_{(i,j)}$}
                \label{fig:placeholder}
            \end{figure}

The values of
$\Delta_{(i,j)}$ can be explicitly computed from the columns of $H$, as stated below.

\begin{lemma}   Given $H=(\bldh_0,\bldh_1,...,\bldh_{n-1})\in\R^{2 \times n}$,
	where the indices are taken modulo $n$,
	then for any two distinct indices $i,j\in[n]$, we have that
\begin{equation}\label{eq:delta-general}
\Delta_{(i,j)}
=
2\sum_{k=0}^{n-1}
\frac{
\|\bldh_k\|
\,\sin\angle (\bldh_k,\bldh_j)
}{
\|\bldh_i\|
\,\sin\angle (\bldh_i,\bldh_j)
}.
\end{equation}
\end{lemma}

\begin{proof}
Note that
\[
\left(e \cdot \bldh_i + S\right)  \cap \left( e' \cdot \bldh_{j} + S\right) = \emptyset
\quad\Longleftrightarrow\quad
e\bldh_i-e'\bldh_j\notin 2S.
\]

Since the set $L_e\triangleq\{\bldx:\bldx=e\bldh_i-e'\bldh_j, e'\in\R\}$ is an affine set,
by Theorem~\ref{thm:separ}, there exists $\bldu\neq0$ such that
\[
\langle\bldu,e\bldh_i-e'\bldh_j\rangle
\ge
b
\ge
\langle\bldu,\blds\rangle,\quad \blds\in2S.
\]
Taking the supremum over $\blds\in 2S$ yields
\[
\langle\bldu,e\bldh_i-e'\bldh_j\rangle
\ge
\sup_{\blds\in 2S} \langle\bldu,\blds\rangle
=
\sup_{|\epsilon_k|\le1}
\left\langle
\bldu,
\sum_{k=0}^{n-1}2\epsilon_k\bldh_k
\right\rangle=
2\sum_{k=0}^{n-1}
|\langle \bldu,\bldh_k\rangle|,
\]
i.e.,
\[
e\langle\bldu,\bldh_i\rangle
-
e'\langle\bldu,\bldh_j\rangle
\ge
2\sum_{k=0}^{n-1}
|\langle\bldu,\bldh_k\rangle|.
\]

Since $e'$ is arbitrary, the above inequality can hold for all
$e'\in\mathbb{R}$ only if
\[
    \langle \bldu,\bldh_j\rangle=0.
\]
Thus, $\bldu$ must be orthogonal to $\bldh_j$. In the two-dimensional case,
this means that $\bldu$ is parallel to a nonzero vector
$\bldh_j^\perp$ perpendicular to $\bldh_j$. Therefore,
\[
    |e|\cdot|\langle \bldh_j^\perp,\bldh_i\rangle|
    \ge
    2\sum_{k=0}^{n-1}
    |\langle \bldh_j^\perp,\bldh_k\rangle|,
\]
i.e.,
\[
|e|
    \ge
    \frac{2\sum_{k=0}^{n-1}
    |\langle \bldh_j^\perp,\bldh_k\rangle|}{|\langle \bldh_j^\perp,\bldh_i\rangle|}
    .
\]
Conversely, suppose that
\[
|e|>
\frac{2\sum_{k=0}^{n-1}|\langle \bldh_j^\perp,\bldh_k\rangle|}
{|\langle \bldh_j^\perp,\bldh_i\rangle|}.
\]
Choose the sign of $\bldh_j^\perp$ so that
$e\langle \bldh_j^\perp,\bldh_i\rangle>0$. For every $e'\in\mathbb R$,
\[
\langle \bldh_j^\perp, e \bldh_i-e' \bldh_j\rangle
=e\langle \bldh_j^\perp,\bldh_i\rangle
>
2\sum_{k=0}^{n-1}|\langle \bldh_j^\perp,\bldh_k\rangle|
=
\sup_{\blds\in 2S}\langle \bldh_j^\perp,\blds\rangle.
\]
Hence $e \bldh_i-e'\bldh_j\notin 2S$ for all $e'\in\mathbb R$.
This proves the claimed formula.
Consequently, 
\[
\begin{aligned}
\Delta_{(i,j)}
=
\frac{
    2\sum_{k=0}^{n-1}
    |\langle \bldh_j^\perp,\bldh_k\rangle|
}{
    |\langle \bldh_j^\perp,\bldh_i\rangle|
}  
=
2\sum_{k=0}^{n-1}
\frac{
    \|\bldh_k\|
    \sin\angle(\bldh_k,\bldh_j)
}{
    \|\bldh_i\|
    \sin\angle(\bldh_i,\bldh_j)
}.
\end{aligned}
\]
\end{proof}

The following lemma establishes a theoretical lower bound on $\Gamma_2(\cC)$ based on the observation that the maximum is no smaller than the average.

\begin{lemma} \label{lem:normalcase} 

Let $n\ge 3$ and $\cC$ be a linear $[n,n-2]$ code over $\R$ with a parity-check matrix $H=(\bldh_0 \ \bldh_1  \ \cdots \ \bldh_{n-1})$.
For each $i \in\Int{n-1}$, 
let $x_i\triangleq\angle (\bldh_i,\bldh_{i+1})$,
and define $x_{n-1}\triangleq\pi+\arg(\bldh_0)-\arg(\bldh_{n-1})$, where the indices are taken modulo $n$. 
Then $x_i>0$,
    $\sum_{i=0}^{n-1}x_i=\pi$,
    $\sin\angle(\bldh_i,\bldh_{i+1})=\sin x_i$ for all $i\in[n]$,
    and 
    \begin{equation}\label{eq:dltbndsine}
        \Gamma_2(\cC)
        \ge
        2\left(
            2+\sum_{p=3}^{n-1}
            \left(
            \prod_{i=0}^{n-1}\frac{\sin\left(\sum_{k=i}^{i+p-2}x_k\right)}{\sin x_i}
            \right)^{\frac{1}{n}}
            \right).
    \end{equation}

    In particular, for $n=3$, $\Gamma_2(\cC)\ge4$.
\end{lemma}

\begin{proof}
According to Proposition~\ref{meanofD} and the definition of $\Delta_{(i,j)}$, we have 
\begin{align}
        \Gamma_2(\cC)
        &=
            \max_{i,j\in[n],i\neq j} \Delta_{(i,j)}\nonumber\\
		  &\ge
		      \max \{\max_{i\in[n]} \Delta_{(i,i+1)}, \max_{{i\in[n]}} \Delta_{(i,i-1)}\}\nonumber\\
	      &\ge
		      \frac{\sum_{i=0}^{n-1}\left(\Delta_{(i,i+1)}+\Delta_{(i,i-1)}\right)}{2n} \nonumber\\
        &=   
            \frac{1}{n}\sum_{i=0}^{n-1}\sum_{j=0}^{n-1}
            \frac{||\bldh_j||\sin\angle (\bldh_j,\bldh_{i+1})}{||\bldh_i||\sin\angle (\bldh_i,\bldh_{i+1})}
            +
            \frac{1}{n}\sum_{i=0}^{n-1}\sum_{j=0}^{n-1}
            \frac{||\bldh_j||\sin\angle (\bldh_j,\bldh_{i-1})}{||\bldh_i||\sin\angle (\bldh_i,\bldh_{i-1})}
            \nonumber\\
        &=
            \frac{1}{n}
            \sum_{p=0}^{n-1}\sum_{i=0}^{n-1}\left(
            \frac{||\bldh_{i+p}||\sin\angle (\bldh_{i+p},\bldh_{i+1})}{||\bldh_{i}||\sin\angle (\bldh_{i},\bldh_{i+1})}
            +
            \frac{||\bldh_{i-p}||\sin\angle (\bldh_{i-p},\bldh_{i-1})}{||\bldh_{i}||\sin\angle (\bldh_{i},\bldh_{i-1})}\right)\nonumber\\
        &=
            \frac{1}{n}
            \left(
            2n+0+\sum_{p=2}^{n-1}\sum_{i=0}^{n-1}
            \left(\frac{||\bldh_{i+p}||\sin\angle (\bldh_{i+p},\bldh_{i+1})}{||\bldh_{i}||\sin\angle (\bldh_{i},\bldh_{i+1})}
            +
            \frac{||\bldh_{i-p}||\sin\angle (\bldh_{i-p},\bldh_{i-1})}{||\bldh_{i}||\sin\angle (\bldh_{i},\bldh_{i-1})}\right)
            \right)
        \nonumber\\
        &\ge
            \frac{1}{n}
            \left(
            2n+\sum_{p=2}^{n-1}2n
            \left(
            \prod_{i=0}^{n-1}\frac{||\bldh_{i+p}||\sin\angle (\bldh_{i+p},\bldh_{i+1})}{||\bldh_{i}||\sin\angle (\bldh_{i},\bldh_{i+1})}\frac{||\bldh_{i-p}||\sin\angle (\bldh_{i-p},\bldh_{i-1})}{||\bldh_{i}||\sin\angle (\bldh_{i},\bldh_{i-1})}
            \right)^{\frac{1}{2n}}
            \right)
            \label{inequality:n=2n}
        \\
        &=
            \frac{1}{n}
            \left(
            2n+\sum_{p=2}^{n-1}2n
            \left(
            \prod_{i=0}^{n-1}\frac{\sin\angle (\bldh_{i+1},\bldh_{i+p})}{\sin\angle (\bldh_{i},\bldh_{i+1})}\frac{\sin\angle (\bldh_{i-p},\bldh_{i-1})}{\sin\angle (\bldh_{i-1},\bldh_{i})}
            \right)^{\frac{1}{2n}}
            \right)
        \nonumber\\
        &=
            2\left(
            1+\sum_{p=2}^{n-1}
            \left(\prod_{i=0}^{n-1}
                \frac
                {\sin\angle (\bldh_{i-p},\bldh_{i-1})}
                {\sin\angle (\bldh_{i},\bldh_{i+1})}
            \right)^{\frac{1}{n}}
            \right)
            \nonumber\\
        &=
            2\left(
            1+1+
            \sum_{p=3}^{n-1}
            \left(\prod_{i=0}^{n-1}
                \frac
                {\sin\angle (\bldh_{i},\bldh_{i+p-1})}
                {\sin\angle (\bldh_{i},\bldh_{i+1})}
            \right)^{\frac{1}{n}}
            \right)\nonumber\\
        &=
            2\left(
            2+\sum_{p=3}^{n-1}
            \left(
            \prod_{i=0}^{n-1}\frac{\sin\left(\sum_{k=i}^{i+p-2}x_k\right)}{\sin x_i}
            \right)^{\frac{1}{n}}
            \right),
            \nonumber
\end{align}
where \eqref{inequality:n=2n} comes from the AM-GM inequality of case  $2n$.
\end{proof}

\begin{remark}\label{remark:neqal3}
    For \(n=3\), Lemma~\ref{lem:normalcase} implies \(\Gamma_2(\mathcal{C}) \ge 4\); in this case the code $\cC(3)$ in \eqref{construction} with $\Gamma_2(\cC)\le \frac{1}{\sin^2(\pi/6)}$ is optimal.
\end{remark}

We will estimate the quantity
$
\left(
\prod_{i=0}^{n-1}
\frac{\sin\left(\sum_{k=i}^{i+p-2}x_k\right)}
{\sin x_i}
\right)^{\frac{1}{n}}.
$
This estimate leads to the desired lower bound.
Its proof, however, is rather involved and relies on a sequence of auxiliary results.
We therefore defer the proof to the next section.

\begin{theorem}[Cyclic Sine-Product Inequality]\label{thm:core}
    Let $n\ge4$.
    Let $x_0,x_1,...,x_{n-1}>0$,
    $\sum_{i=0}^{n-1} x_i=\pi$,
    where the indices are taken modulo $n$.
    For $0\le p\le n$, define $\cX_i^{(p)}\triangleq x_i+x_{i+1}+\cdots+x_{i+p-1}$,
where $\cX_i^{(0)}=0$ and $\cX_i^{(n)}=\pi$. Then, for every $0\le p\le n$,
\begin{equation}\label{eq:core}
    \left(\prod_{i=0}^{n-1} \frac{\sin\cX_i^{(p)}}{\sin x_i}\right)^{\frac{1}{n}}
\ge
    \frac{\sin(p\pi/n)}{\sin(\pi/n)}.
\end{equation}
\end{theorem}

Our lower bound for the analog code also relies on the following auxiliary lemma.
    \begin{lemma}\label{lem:simplify}
    For any integer $n \ge 2$, we have
        \begin{equation}\label{sinsumdivi}
            \frac{\sum_{i=1}^{n-1}\sin \frac{i\pi}{n}}{\sin\frac{\pi}{n}}=\frac{1}{2\sin^2\frac{\pi}{2n}}.
        \end{equation}
    \end{lemma}

    \begin{proof}
    By Lagrange's sine identity \cite[Page 129]{Jeffrey1995HandbookOM},
    we obtain
    \[
    \sum_{i=1}^m \sin(ix)
        = \frac{\sin\frac{(m+1)x}{2}\sin\frac{mx}{2}}{\sin\frac{x}{2}}.
    \]
    
    Let $x=\frac{\pi}{n}$ and $m=n-1$, then
    \[
    \sum_{i=1}^{n-1} \sin\frac{i\pi}{n}
    = \frac{\sin\frac{\pi}{2}\sin\frac{(n-1)\pi}{2n}}{\sin\frac{\pi}{2n}}
    = \frac{\cos\frac{\pi}{2n}}{\sin\frac{\pi}{2n}}.
    \]

Therefore,
\[
\frac{\sum_{i=1}^{n-1} \sin \frac{i\pi}{n}}{\sin\frac{\pi}{n}}
= \frac{\cos\frac{\pi}{2n}}{2\sin^2\frac{\pi}{2n}\cos\frac{\pi}{2n}}
= \frac{1}{2\sin^2\frac{\pi}{2n}}.
\]
\end{proof}

	\begin{theorem}
		Let $n\ge3$ and let $\cC$ be a linear $[n,n-2]$ code over $\R$. 
		We have
		\begin{equation}
		\Gamma_2(\cC)
		\ge
		\frac{1}{\sin^2\frac{\pi}{2n}}.
		\end{equation}
	\end{theorem}
	\begin{proof}
    By Remark~\ref{remark:neqal3}, the theorem holds for \(n=3\).
    Thus, in the sequel we assume that \(n\ge 4\).
        \begin{align}
            \Gamma_2(\cC)
        &\overset{\eqref{eq:dltbndsine}}{\ge}
            2\left(
            2+\sum_{p=3}^{n-1}
            \left(
            \prod_{i=0}^{n-1}\frac{\sin\left(\sum_{k=i}^{i+p-2}x_k\right)}{\sin x_i}
            \right)^{\frac{1}{n}}
            \right)\nonumber\\
        &\overset{\eqref{eq:core}}{\ge}
            2\left(
            2+\sum_{p=3}^{n-1}\frac{\sin\frac{(p-1)\pi}{n}}{\sin\frac{\pi}{n}}
            \right)\nonumber\\
        &=
            2
            \sum_{i=1}^{n-1}\frac{\sin\frac{i\pi}{n}}{\sin\frac{\pi}{n}}\nonumber\\
        &\overset{\eqref{sinsumdivi}}{=}
            \frac{1}{\sin^2\frac{\pi}{2n}}.\nonumber
        \end{align}
	\end{proof} 

    Therefore, Roth's construction is optimal.
\section{Proof of Cyclic Sine Product Inequality}\label{sec:proof}
Our proof requires the following auxiliary lemmas.
\begin{lemma}\label{lem:minkowskitypeinequality}
If $A_i,B_i \ge 0$ for $i=0,\ldots,n-1$, then
\begin{equation}\label{eq:productinequality}
    \left(\prod_{i=0}^{n-1} (A_i+B_i)\right)^{1/n}
\ge
\left(\prod_{i=0}^{n-1} A_i\right)^{1/n}
+
\left(\prod_{i=0}^{n-1} B_i\right)^{1/n}.
\end{equation}
\end{lemma}
\begin{proof}
If some $A_i=0$ or some $B_i=0$, the claim follows immediately. 
Thus we may assume that $A_i>0$ and $B_i>0$ for all $i$.

From AM-GM inequality,
we obtain 
\begin{equation}\label{eq:Ahalf}
    \frac{1}{n}\sum_{i=0}^{n-1}\frac{A_i}{A_i+B_i}
\ge
    \left(\prod_{i=0}^{n-1}\frac{A_i}{A_i+B_i}\right)^{\frac{1}{n}},
\end{equation}
and
\begin{equation}\label{eq:Bhalf}
    \frac{1}{n}\sum_{i=0}^{n-1}\frac{B_i}{A_i+B_i}
\ge
    \left(\prod_{i=0}^{n-1}\frac{B_i}{A_i+B_i}\right)^{\frac{1}{n}}.
\end{equation}
Adding Equation~\eqref{eq:Ahalf} and \eqref{eq:Bhalf} gives
\[
    1
\ge
    \left(\prod_{i=0}^{n-1}\frac{A_i}{A_i+B_i}\right)^{\frac{1}{n}}
    +\left(\prod_{i=0}^{n-1}\frac{B_i}{A_i+B_i}\right)^{\frac{1}{n}},
\]
i.e.,
\[
\left(\prod_{i=0}^{n-1} (A_i+B_i)\right)^{1/n}
\ge
\left(\prod_{i=0}^{n-1} A_i\right)^{1/n}
+
\left(\prod_{i=0}^{n-1} B_i\right)^{1/n}.
\]
\end{proof}

\begin{lemma}
   For all real numbers $\alpha,\beta,\gamma,a,b$,  we have
\begin{equation}\label{eq:sinid}
    \sin(\alpha+\beta)\sin(\beta+\gamma)
=
\sin \alpha\sin \gamma+
\sin \beta\sin(\alpha+\beta+\gamma),
\end{equation}
and 
\begin{equation}\label{eq:elementray}
    \sin^2 a-\sin(a-b)\sin(a+b)=\sin^2b.
\end{equation}
\end{lemma}
\begin{proof}
    By the product-to-sum formula,
\[
\sin(\alpha+\beta)\sin(\beta+\gamma)
=
\frac{\cos(\alpha-\gamma)-\cos(\alpha+2\beta+\gamma)}{2}
\]
and 
\[
\sin \beta\sin(\alpha+\beta+\gamma)
=
\frac{\cos(\alpha+\gamma)-\cos(\alpha+2\beta+\gamma)}{2}.
\]
Subtracting the second identity from the first gives
\[
\sin(\alpha+\beta)\sin(\beta+\gamma)-\sin \beta\sin(\alpha+\beta+\gamma)
=
\frac{\cos(\alpha-\gamma)-\cos(\alpha+\gamma)}{2}.
\]
Since
\[
\cos(\alpha-\gamma)-\cos(\alpha+\gamma)=2\sin \alpha\sin \gamma,
\]
we obtain \eqref{eq:sinid}.

For the second identity, using
\[
\sin(a-b)\sin(a+b)
=
\frac{\cos(2b)-\cos(2a)}{2},
\]
we get
\[
\sin^2 a-\sin(a-b)\sin(a+b)
=
\frac{1-\cos(2a)}{2}
-
\frac{\cos(2b)-\cos(2a)}{2}
=
\frac{1-\cos(2b)}{2}
=
\sin^2 b.
\]
This proves \eqref{eq:elementray}.
\end{proof}

\begin{proof}[Proof of Theorem~\ref{thm:core}]
For $0\le p\le n$, define
\[
D_p\triangleq\left(\prod_{i=0}^{n-1} \sin\cX_i^{(p)}\right)^{1/n}.
\]
Thus $D_0=0,
D_n=0,$
because $\sin\cX_i^{(0)}=0$ and $\sin\cX_i^{(n)}=\sin\pi=0$.

Fix $1\le p\le n-1$. For each cyclic index $i$, set
\[
\alpha\triangleq x_i,
\qquad
\beta\triangleq\cX_{i+1}^{(p-1)},
\qquad
\gamma\triangleq x_{i+p}.
\]
Then
\[
\alpha+\beta=\cX_i^{(p)},
\qquad
\beta+\gamma=\cX_{i+1}^{(p)},
\qquad
\alpha+\beta+\gamma=\cX_i^{(p+1)}.
\]
Applying \eqref{eq:sinid} gives
\begin{equation}\label{eq:recur}
    \sin\cX_i^{(p)}\sin\cX_{i+1}^{(p)}
=
\sin x_i\sin x_{i+p}
+
\sin\cX_{i+1}^{(p-1)}\sin\cX_i^{(p+1)}.
\end{equation}

Define
\[
A_i\triangleq\sin x_i\sin x_{i+p},
\qquad
B_i\triangleq\sin\cX_{i+1}^{(p-1)}\sin\cX_i^{(p+1)}.
\]
Then \eqref{eq:recur} says
\[
\sin\cX_i^{(p)}\sin\cX_{i+1}^{(p)}=A_i+B_i.
\]
Multiplying over $i=0,\ldots,n-1$, we get
\[
\prod_{i=0}^{n-1} \sin\cX_i^{(p)}\sin\cX_{i+1}^{(p)}
=
\left(\prod_{i=0}^{n-1} \sin\cX_i^{(p)}\right)^2
=D_p^{2n}.
\]
By Lemma~\ref{lem:minkowskitypeinequality},
\begin{equation}\label{eq:afterproductinequality}
    D_p^2
=
\left(\prod_{i=0}^{n-1}(A_i+B_i)\right)^{1/n}
\ge
\left(\prod_{i=0}^{n-1}A_i\right)^{1/n}
+
\left(\prod_{i=0}^{n-1}B_i\right)^{1/n}.
\end{equation}
Now
\[
\prod_{i=0}^{n-1}A_i
=
\prod_{i=0}^{n-1}\sin x_i\sin x_{i+p}
=
\left(\prod_{i=0}^{n-1}\sin x_i\right)^2
=D_1^{2n},
\]
Similarly,
\[
\prod_{i=0}^{n-1}B_i
=
\prod_{i=0}^{n-1}\sin\cX_{i+1}^{(p-1)}\sin\cX_i^{(p+1)}
=D_{p-1}^nD_{p+1}^n.
\]
Substituting these two identities into \eqref{eq:afterproductinequality} yields
\begin{equation}\label{eq:finalrecur}
    D_p^2\ge D_1^2+D_{p-1}D_{p+1}
\qquad (1\le p\le n-1).
\end{equation}

Since $0<x_i<\pi$, we have $D_1>0$. Define
\[
X_p\triangleq\frac{D_p}{D_1}
\qquad (0\le p\le n).
\]
Dividing \eqref{eq:finalrecur} by $D_1^2$ gives
\begin{equation}\label{eq:normalizedrecur}
    X_p^2\ge 1+X_{p-1}X_{p+1}
\qquad (1\le p\le n-1).
\end{equation}
Also $X_0=0,X_n=0,X_1=1.$
Finally, since $\cX_i^{(n-1)}=\pi-x_{i-1}$,
we have $\sin\cX_i^{(n-1)}=\sin x_{i-1}.$
Therefore $D_{n-1}=D_1$ and $X_{n-1}=1$.

Set
\[
b\triangleq\frac{\pi}{n},
\qquad
S_p\triangleq\frac{\sin(pb)}{\sin b}
\qquad (1\le p\le n-1).
\]
Then $S_0=0$, $S_n=0$, $S_1=1$ and $S_{n-1}=1$.
Apply Equation~\eqref{eq:elementray}
with $a=pb$ gives
\[
\sin^2(pb)-\sin((p-1)b)\sin((p+1)b)=\sin^2 b.
\]
Dividing by $\sin^2 b$ gives
\begin{equation}\label{eq:recurofS}
    S_p^2=1+S_{p-1}S_{p+1}
\qquad (1\le p\le n-1).
\end{equation}

We claim that
\begin{equation}\label{eq:compareofXandS}
    X_p\ge S_p
\qquad
(1\le p\le n-1).
\end{equation}
Assume, toward a contradiction, that \eqref{eq:compareofXandS} fails. Since $S_p>0$ for $1\le p\le n-1$, define
\[
q\triangleq\min_{1\le p\le n-1}\frac{X_p}{S_p}.
\]

The failure of \eqref{eq:compareofXandS} means $q<1$.
Choose $p$ such that $\frac{X_p}{S_p}=q$.
By the definition of $q$,
we obtain
\begin{align}\label{inequality:lowerboundofX_i}
    X_{p-1}\ge qS_{p-1},
\qquad
X_{p+1}\ge qS_{p+1},
\qquad
X_p=qS_p.
\end{align}
Hence,
we derive
\begin{align}
    q^2S_p^2
\overset{\eqref{eq:normalizedrecur}}{\ge}
    1+X_{p-1}X_{p+1}
\overset{\eqref{inequality:lowerboundofX_i}}{\ge}
    1+q^2S_{p-1}S_{p+1}
\overset{\eqref{eq:recurofS}}{=}
    1+q^2(S_p^2-1)\nonumber
\end{align}
i.e., $q^2\ge1$, contradicting $0\le q<1$. Hence the claim \eqref{eq:compareofXandS} is proved.

    And Equation~\eqref{eq:compareofXandS} yields
    \[
        \left(\prod_{i=0}^{n-1}\frac{\sin\cX_i^{(p)}}{\sin x_i}\right)^{\frac{1}{n}}
    =
        \frac{D_p}{D_1}
    =
        X_p
    \ge 
        S_p
    =
        \frac{\sin(p\pi/n)}{\sin(\pi/n)}.
    \]
\end{proof}
	
\section{Upper bounds of higher redundancy}\label{sec:upperbound}
We now study upper bounds on $\Gamma_2(\mathcal{C})$ for analog codes with redundancy $r\ge 2$. We first prove an existence result showing that there is a real linear $[n,\ge n-r]$ code $\mathcal{C}$ satisfying $ \Gamma_2(\mathcal{C}) \le O\left(n^{1+\frac{1}{r-1}}\right).$ We then give an explicit construction attaining the same exponent, albeit with a larger constant. By Proposition~\ref{prop:coh}, deriving upper bounds for larger redundancy reduces to constructing $n$ unit vectors in $\mathbb{R}^r$ with small coherence.

Let $d\eqdef r-1$. The unit sphere in $\R^r$ is
\[
\bbS^d\eqdef\{\bldx\in \R^{d+1}:\norm{\bldx}=1\}.
\]
Denote
\[
|\bbS^d|\triangleq\frac{2\pi^{(d+1)/2}}{\Gamma((d+1)/2)}
\]
for the surface area of the unit sphere $\bbS^d$, 
and denote 
\begin{equation}\label{eq:cd}
    c_d\triangleq
\left(\frac{d|\bbS^d|}{|\bbS^{d-1}|}\right)^{1/d}.
\end{equation}
For $\bldx \in \bbS^d$ and angular radius $0<\alpha \le \pi$, the spherical cap of radius $\alpha$ centered at $\bldx$ is
\[
\text{Cap}(\bldx,\alpha) = \{\bldy \in \bbS^d : \arccos \ip{\bldx}{\bldy} \le \alpha, \bldx\in\bbS^d \}.
\]
We denote the area of $\text{Cap}(\bldx,\alpha)$ by $|\text{Cap}(\bldx,\alpha)|$.
Then, for any $\bldx\in\mathbb{S}^d$, we have \[ |\operatorname{Cap}(\bldx,\alpha)| = |\mathbb{S}^{d-1}| \int_0^\alpha \sin^{d-1}\theta\,d\theta; \] see, e.g., \cite{Li2011HypersphericalCap}.

\begin{lemma}\label{lem:projective-packing}
Let $d\ge 1$. 
Then, for every sufficiently large integer $n$, there exist points
\[
\bldx_0,\ldots,\bldx_{n-1}\in \bbS^d,
\]
satisfying
\begin{equation}\label{eq:sep}
\arccos\langle\bldx_i,\bldx_j\rangle
\ge c_d n^{-1/d},
\qquad i\ne j.
\end{equation}
\end{lemma}

\begin{proof} Let $\alpha=c_d n^{-1/d}$.
For all sufficiently large $n$, we have $0<\alpha\le 1$.
Let
\[
\mathscr P=\{\bldx_0,\ldots,\bldx_{M-1}\}\subseteq \bbS^d
\]
be a maximal $\alpha$-separated set, meaning that $\arccos\langle\bldx_i,\bldx_j\rangle\ge \alpha$,
for all $i\ne j$, and no further point of $\bbS^d$ can be added while preserving this property.

By maximality, the spherical caps $\operatorname{Cap}(\bldx_i,\alpha)$, $0\le i\le M-1$, cover $\mathbb{S}^d$.
Hence, \begin{equation*} \begin{aligned} |\mathbb{S}^{d}| &\le \sum_{i=0}^{M-1} |\operatorname{Cap}(\bldx_i,\alpha)| = M|\mathbb{S}^{d-1}| \int_0^\alpha \sin^{d-1}\theta\,d\theta \\ &\le M|\mathbb{S}^{d-1}| \int_0^\alpha \theta^{d-1}\,d\theta = M\frac{|\mathbb{S}^{d-1}|}{d}\alpha^d . \end{aligned} \end{equation*}
Therefore,
\[
M\ge \frac{d|\bbS^d|}{|\bbS^{d-1}|\alpha^d}=n.
\]
\end{proof}

\begin{theorem}[Existence bound]\label{thm:existence-explicit}
Fix $r\ge 2$. For every sufficiently large $n$, there exists a real linear $[n,\ge n-r]$ code $\cC$ such that
\begin{equation}\label{eq:existence-explicit}
\Gamma_2(\cC)
\le
\frac{2n}{\sin\left(c_{r-1}n^{-1/(r-1)}\right)}
\le
\frac{\pi}{c_{r-1}}\,n^{1+\frac{1}{r-1}},
\end{equation}
where $c_{r-1}$ is defined as in \eqref{eq:cd},
and only depends on $r$.
\end{theorem}

\begin{proof}
By Lemma~\ref{lem:projective-packing}, there are points $\bldh_0,\ldots,\bldh_{n-1}\in \bbS^{r-1}$
such that
\[
\arccos\langle\bldh_i,\bldh_j\rangle
\ge c_{r-1}n^{-1/(r-1)}
\]
for all $i\ne j$.
Let $H=(\bldh_0,\ldots,\bldh_{n-1})\in \R^{r\times n}$.
Let $\cC$ be the linear code over $\R$ with parity-check matrix $H$.
Then
\[
\rho(H)\le \cos\left(c_{r-1}n^{-1/(r-1)}\right).
\]
By \eqref{eq:coherence-bound},
\begin{equation*}\label{eq:sinfangsuo}
    \Gamma_2(\cC)
\le
\frac{2n}{\sqrt{1-\rho(H)^2}}
\le
\frac{2n}{\sin\left(c_{r-1}n^{-1/(r-1)}\right)}
\le
\frac{\pi}{c_{r-1}}n^{1+\frac{1}{r-1}}
.
\end{equation*}
The last inequality follows from the standard estimate  $\sin x\ge \frac{2}{\pi}x$ when $0<x\le \frac{\pi}{2}$. For all sufficiently large $n$, we have $ c_{r-1}n^{-1/(r-1)}\in (0,\pi/2], $ so the estimate applies. This completes the proof.
\end{proof}

We now give an explicit construction achieving the same exponent, although generally with a worse constant.
Our construction takes lattice points in a ball centered at the origin.

Specifically,
for $d\ge 1$, let
\[
\mathcal B_d(0,T)\triangleq\{\bldx\in\R^d:\|\bldx\|\le T\}
\]
and define
\[
N_d(T)\triangleq|\Z^d\cap \mathcal B_d(0,T)|.
\]
Let
\[
\kappa_d\triangleq|\mathcal B_d(0,1)|=\frac{\pi^{d/2}}{\Gamma(d/2+1)}.
\]
The following lemma describes how large $T$ needs to be
for $\mathcal B_d(0,T)$ to contain at least $n$ points of $\Z^d$.
\begin{lemma}
\label{lem:lattice-count}
For every $T\ge \sqrt d/2$,
\[
N_d(T)\ge \kappa_d\left(T-\frac{\sqrt d}{2}\right)^d.
\]
Consequently, for every $n\ge 1$, if
\[
T_n=\left\lceil \left(\frac{n}{\kappa_d}\right)^{1/d}+\frac{\sqrt d}{2}\right\rceil,
\]
then
\[
N_d(T_n)\ge n.
\]
\end{lemma}

\begin{proof}
For each $\bldm\in\Z^d$, let
\[
Q_{\bldm}=\bldm+[-1/2,1/2]^d
\]
be the unit cube centered at $\bldm$. If $\bldx\in B_d(0,T-\sqrt d/2)$ and $\bldm\in\Z^d$ is a nearest lattice point to $\bldx$ in the $\infty$-norm sense, then $\bldx\in Q_{\bldm}$ and
\[
\|\bldm\|\le \|\bldx\|+\|\bldm-\bldx\|\le T-\frac{\sqrt d}{2}+\frac{\sqrt d}{2}=T.
\]
Therefore
\[
\mathcal B_d(0,T-\sqrt d/2)
\subseteq
\bigcup_{\bldm\in \Z^d\cap B_d(0,T)}Q_\bldm.
\]
Taking volumes gives
\[
\kappa_d\left(T-\frac{\sqrt d}{2}\right)^d
\le
\sum_{\bldm\in \Z^d\cap B_d(0,T)}\vol(Q_\bldm)
=N_d(T),
\]
since each cube has volume $1$. The second assertion follows by substituting $T=T_n$. 
\end{proof}

\begin{construction}[Ball-grid parity-check matrix]
\label{cons:ball-grid}
Fix $r\ge 2$ and set
\[
d=r-1,\quad
T_n=\left\lceil \left(\frac{n}{\kappa_d}\right)^{1/d}+\frac{\sqrt d}{2}\right\rceil.
\]
Choose $n$ distinct lattice points
\[
\bldm_0,\ldots,\bldm_{n-1}\in \Z^d\cap \mathcal B_d(0,T_n).
\]
For each $0\le i \le n-1$, define
\[
\bldy_i=\frac{\bldm_i}{T_n}\in \mathcal B_d(0,1)
\]
and
\[
\bldh_i=\frac{(1,\bldy_i)}{\sqrt{1+\|\bldy_i\|^2}}\in\Sph^{r-1}\subset\R^r.
\]
Let
\[
H=(\bldh_0,\ldots,\bldh_{n-1})\in\R^{r\times n},
\]
and 
let $\cC$ be the linear code over $\R$ with parity-check matrix $H$.
\end{construction}

\begin{theorem}
    Fix $r\ge 2$.
    For every $n\ge r$,
    Construction~1 gives a linear $[n,\ge n-r]$ code $\cC$ over $\R$ satisfying
    \[
    \Gamma_2(\cC)
\le
    4n\left\lceil
    \left(
    \frac{n}{\kappa_{r-1}}
    \right)^{\frac{1}{r-1}}+\frac{\sqrt{r-1}}{2
}
    \right\rceil
    =
    O\left(
    n^{1+\frac{1}{r-1}}
    \right).
    \]
\end{theorem}
\begin{proof}
For two distinct indices $i\neq j$, 
we obtain
\[
\|\bldy_i-\bldy_j\|
=
\frac{\|\bldm_i-\bldm_j\|}{T_n}
\ge
\frac{1}{T_n}.
\]
Moreover,
\[
\langle \bldh_i,\bldh_j\rangle
=
\frac{1+\langle \bldy_i,\bldy_j\rangle}
{\sqrt{1+\|\bldy_i\|^2}\sqrt{1+\|\bldy_j\|^2}}.
\]
Hence
\[
\begin{aligned}
1-|\langle \bldh_i,\bldh_j\rangle|^2
&=
\frac{
(1+\|\bldy_i\|^2)(1+\|\bldy_j\|^2)
-
(1+\langle \bldy_i,\bldy_j\rangle)^2
}
{(1+\|\bldy_i\|^2)(1+\|\bldy_j\|^2)}  \\
&=
\frac{
\|\bldy_i-\bldy_j\|^2
+
\|\bldy_i\|^2\|\bldy_j\|^2-\langle \bldy_i,\bldy_j\rangle^2
}
{(1+\|\bldy_i\|^2)(1+\|\bldy_j\|^2)}  \\
&\ge
\frac{\|\bldy_i-\bldy_j\|^2}{4},
\end{aligned}
\]
where we used the Cauchy--Schwarz inequality and the fact that
$\|\bldy_i\|,\|\bldy_j\|\le 1$. Therefore,
\[
\sqrt{1-|\langle \bldh_i,\bldh_j\rangle|^2}
\ge
\frac{\|\bldy_i-\bldy_j\|}{2}
\ge
\frac{1}{2T_n}.
\]
Taking the minimum over all distinct pairs gives
\[
\sqrt{1-\rho(H)^2}\ge \frac{1}{2T_n}.
\]
By the bound in \eqref{eq:coherence-bound},
\[
\Gamma_2(\cC)\le
\frac{2n}{\sqrt{1-\rho(H)^2}}
\le
4nT_n
=
4n
\left\lceil
\left(\frac{n}{\kappa_{r-1}}\right)^{1/(r-1)}
+
\frac{\sqrt{r-1}}{2}
\right\rceil.
\]
\end{proof}

Figure~\ref{fig:comparison} compares the leading terms of the two upper bounds obtained from Theorem~\ref{thm:existence-explicit} and Construction~\ref{cons:ball-grid}. To make the comparison visually clearer, the curve corresponding to the existence upper bound is plotted after an appropriate rescaling. Figure~\ref{fig:r=2compare} compares the two upper bounds derived in this paper with Roth's construction in the case $r=2$, while Figure~\ref{fig:r=3compare} compares them with Song and Cai's construction in the case $r=3$.

\begin{figure}[H]
        \centering
        \includegraphics[width=0.5\linewidth]{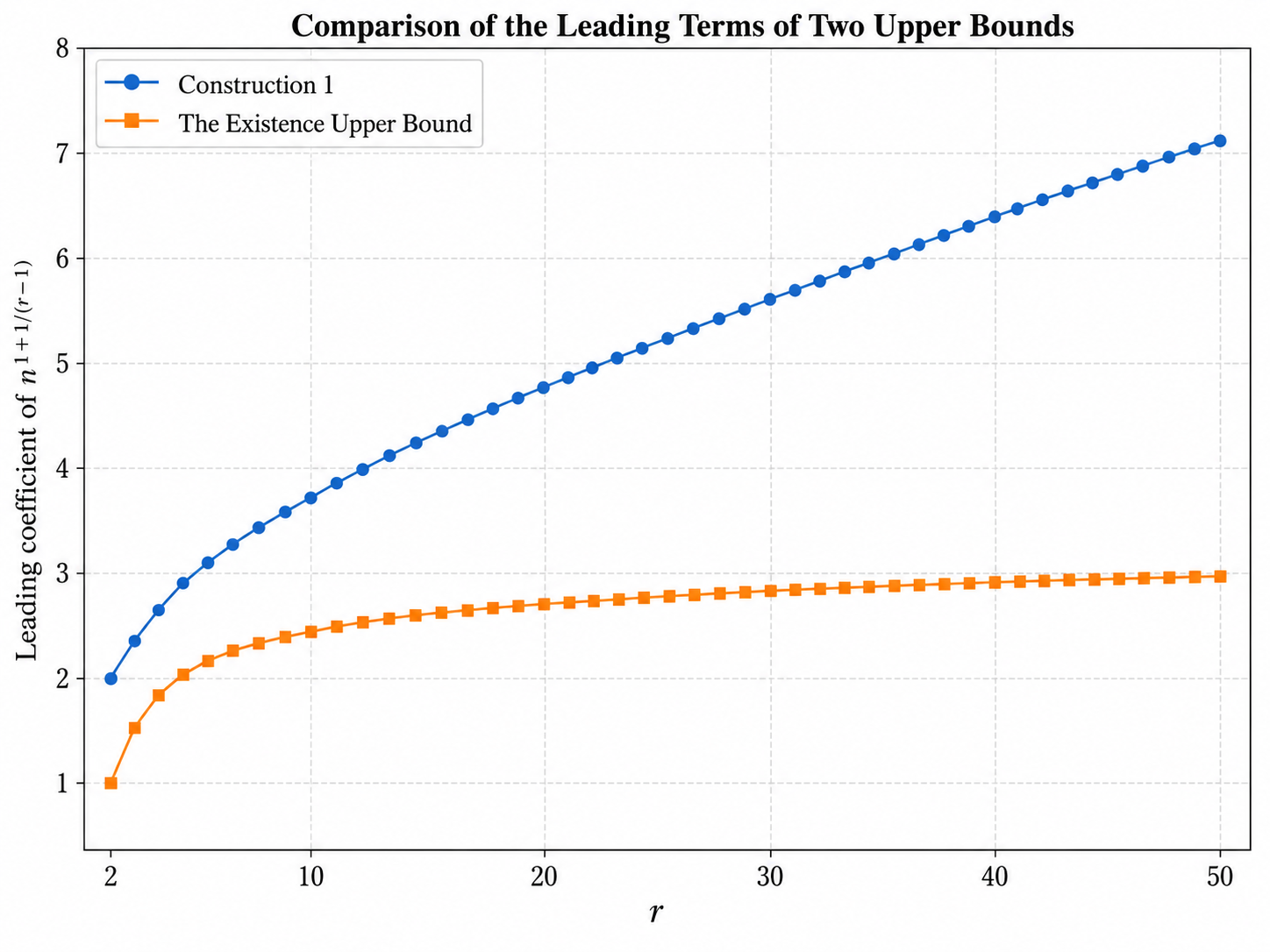}
                \centering\caption{Comparison of Leading Terms in the Two Upper Bounds}
                \label{fig:comparison}
            \end{figure}
    
    \begin{figure}[htbp]
	\centering
	\begin{minipage}{0.49\linewidth}
		\centering
		\includegraphics[width=0.9\linewidth]{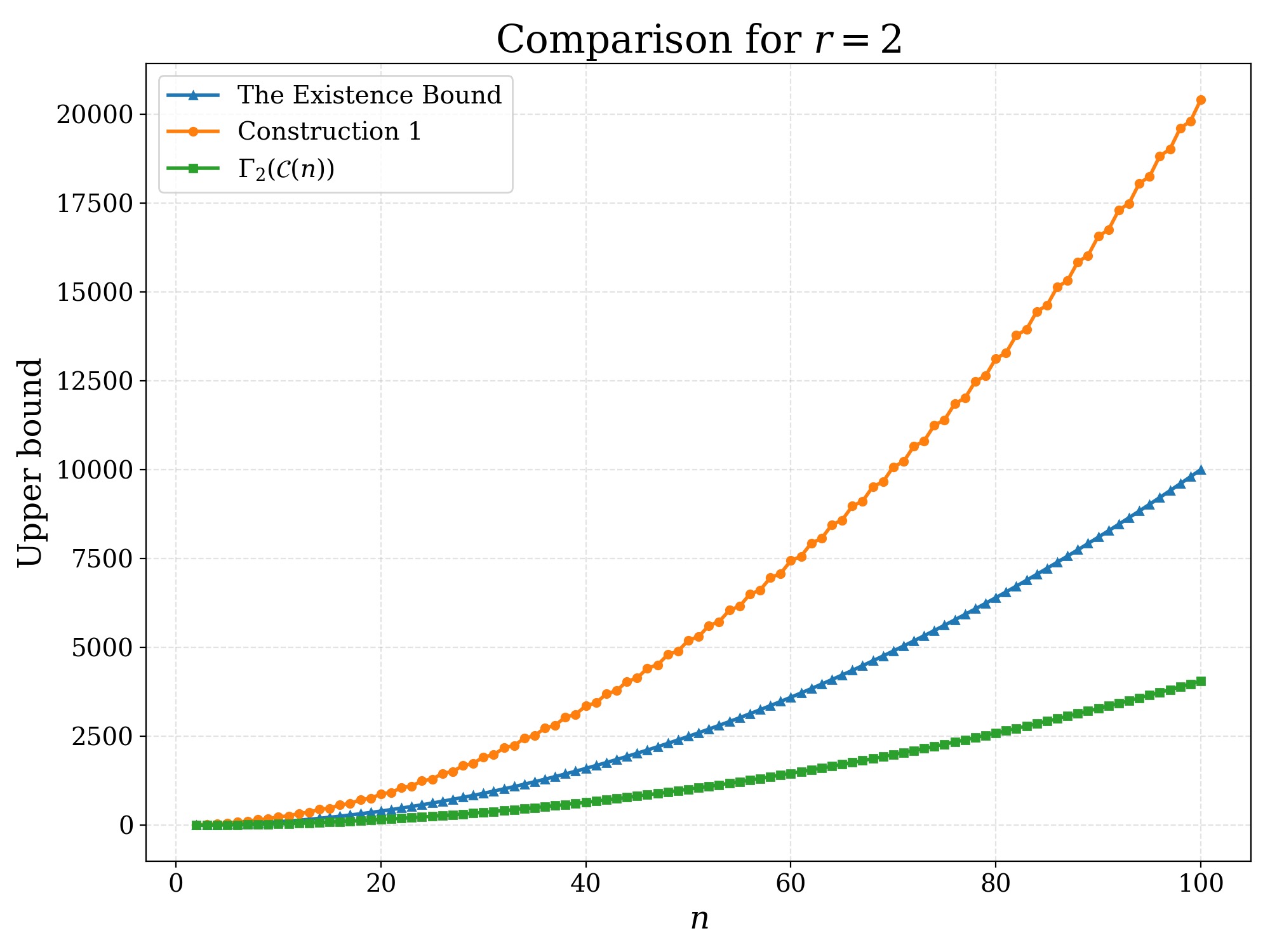}
		\caption{Comparison of Upper Bounds for $r=2$}
		\label{fig:r=2compare}
	\end{minipage}
	\begin{minipage}{0.49\linewidth}
		\centering
		\includegraphics[width=0.9\linewidth]{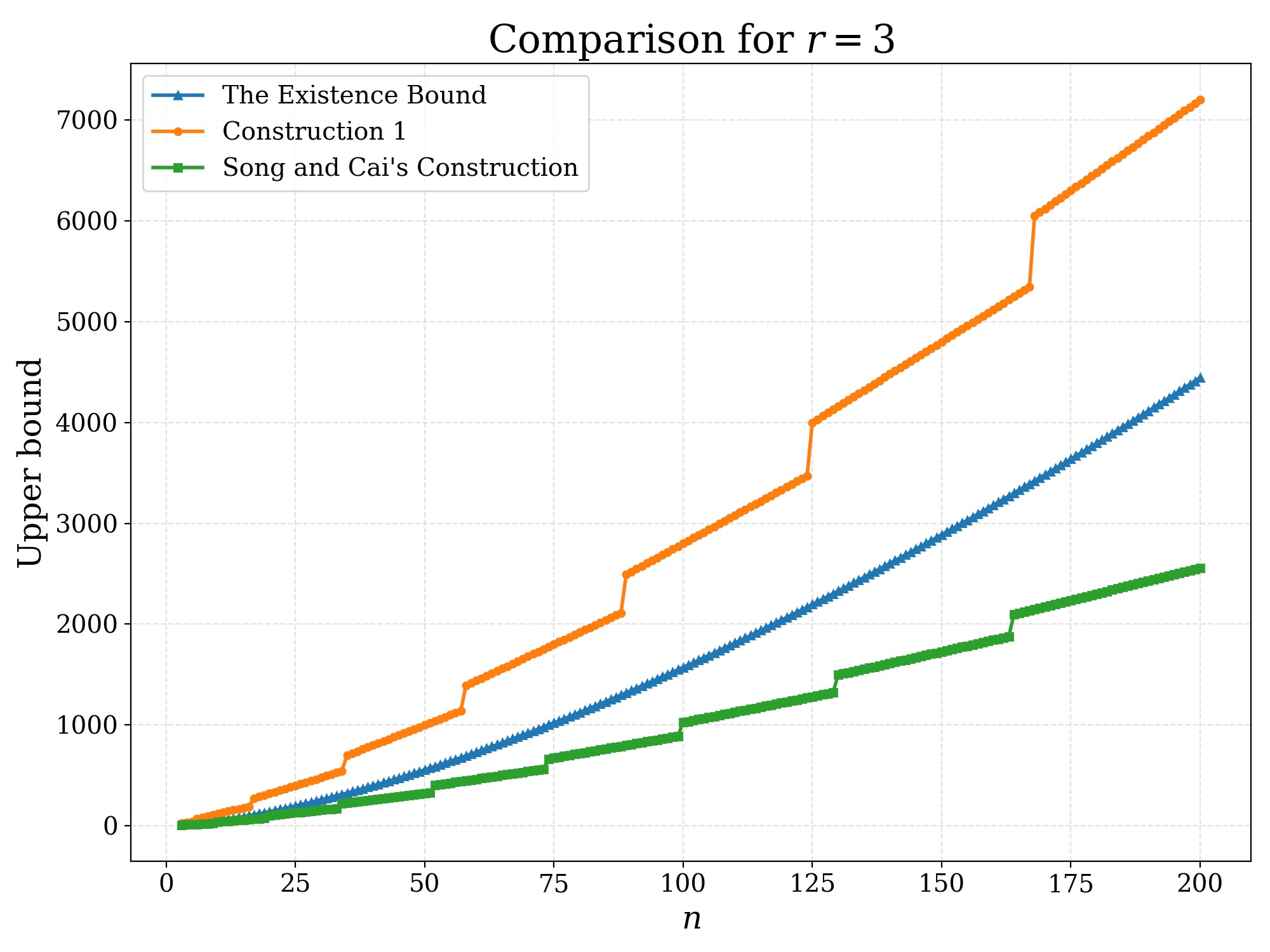}
		\caption{Comparison of Upper Bounds for $r=3$}
		\label{fig:r=3compare}
	\end{minipage}
\end{figure}
\FloatBarrier

\bibliographystyle{IEEEtranS}
\bibliography{allbib}
\end{document}